\documentclass[a4paper,11pt]{article}
\usepackage{graphicx} 

\usepackage[margin=3cm,bottom=2.5cm]{geometry}
\usepackage{hyperref}
\usepackage[usenames,dvipsnames]{color}
\usepackage{amssymb,amsmath,graphics,color,enumerate}
\usepackage{tikz}
\usepackage{xspace}
\usetikzlibrary{arrows,shapes,decorations,automata,backgrounds,petri,patterns} 
\usepackage{todonotes}
\usepackage[utf8x,utf8]{inputenc} 
\usepackage[T1]{fontenc}
\usepackage{cleveref}
\usepackage{booktabs}
\usepackage{array}
\usepackage{mathtools}
\usepackage{algorithm}
\usepackage{caption}
\usepackage[noend]{algpseudocode}
\usepackage[numbers,sort]{natbib}
\usepackage{standalone}
\usepackage{tikz}
\usetikzlibrary{arrows}
\usetikzlibrary{patterns}

\usepackage{epsfig}
\usepackage{subfig}
\usepackage{amsmath,amsfonts,amssymb,epsfig,color,amsthm}
\usepackage{comment}
\usepackage{thmtools}
\usepackage{thm-restate}
\usepackage{mdframed}
\usepackage[absolute]{textpos}
\usepackage{enumitem}
\usepackage{soul}

\newtheorem{theorem}{Theorem}[section]
\newtheorem{definition}[theorem]{Definition}
\newtheorem{lemma}[theorem]{Lemma}
\newtheorem{corollary}[theorem]{Corollary}

\newtheorem{observation}[theorem]{Observation}

\newtheorem{conjecture}{Conjecture}

\newcommand{\Tt}{{\ensuremath{\mathcal{T}}}}

\newcommand{\Oh}{{\ensuremath{\mathcal{O}}}}

\newcommand{\anonymyze}[1]{}%

\newcommand{\td}{\mathsf{td}}

\newcommand{\floor}[1]{\ensuremath{\left\lfloor{#1}\right\rfloor}}%
\newcommand{\ignore}[1]{}%

\newcommand{\ProblemFormat}[1]{\textsc{#1}}
\newcommand{\ProblemName}[1]{\ProblemFormat{#1}\xspace}

\newcommand{\MinHFunc}{\mathsf{MinHLT}}

\newcommand{\MinHProb}{\ProblemName{MinHLT}}

\newcommand{\DualMinHProb}{\ProblemName{DualMinHLT}}

\newcommand\MinLLTProb{\ProblemName{MinLLT}}
\newcommand\MaxLLTProb{\ProblemName{MaxLLT}}

\newcommand{\Compr}{\mathrm{Compr}}

\usepackage{mathrsfs}  
\usepackage{boxedminipage}
\usepackage{framed}
\usepackage{xifthen}
\usepackage{tabularx}
\usepackage{tikz} 
\usetikzlibrary{calc}

\newlength{\RoundedBoxWidth}
\newsavebox{\GrayRoundedBox}
\newenvironment{GrayBox}[1]%
   {\setlength{\RoundedBoxWidth}{.93\textwidth}
    \def\boxheading{#1}
    \begin{lrbox}{\GrayRoundedBox}
       \begin{minipage}{\RoundedBoxWidth}}%
   {   \end{minipage}
    \end{lrbox}
    \begin{center}
    \begin{tikzpicture}%
       \node(Text)[draw=black!20,fill=white,rounded corners,%
             inner sep=2ex,text width=\RoundedBoxWidth]%
             {\usebox{\GrayRoundedBox}};
        \coordinate(x) at (current bounding box.north west);
        \node [draw=white,rectangle,inner sep=3pt,anchor=north west,fill=white] 
        at ($(x)+(6pt,.75em)$) {\boxheading};
    \end{tikzpicture}
    \end{center}}     

\newenvironment{defproblemx}[2][]{\noindent\ignorespaces%
                                \FrameSep=6pt%
                                \parindent=0pt%
                \ifthenelse{\isempty{#1}}{%
                  \begin{GrayBox}{\textsc{#2}}%
                }{%
                  \begin{GrayBox}{\textsc{#2} parameterized by~{#1}}%
                }
                \begin{tabular*}{\textwidth}{@{\hspace{.1em}} >{\itshape} p{1.8cm} p{0.8\textwidth} @{}}%
            }{
                \end{tabular*}%
                \end{GrayBox}%
                \ignorespacesafterend
            }  

\newcommand{\fancyparproblemdef}[4]{
  \begin{defproblemx}[#3]{#1}
    Input:  & #2 \\
    Question: & #4
  \end{defproblemx}
}%

\newcommand{\fancyproblemdef}[3]{
  \begin{defproblemx}{#1}
    Input:  & #2 \\
    Question: & #3
  \end{defproblemx}
}%

\usepackage{xspace}
\newcommand\mso[1]{\ensuremath{\mathsf{#1}}}
\newcommand\MSOone{\ensuremath{\mathsf{MSO}_1}\xspace}
\newcommand\MSOtwo{\ensuremath{\mathsf{MSO}_2}\xspace}

\newcommand\bags{\beta}

\newcommand\tw{\mathsf{tw}}
\newcommand\const{\mathrm{DFS}}

\newcommand\roots{R}

\newcommand\subparagraphstar[1]{\subparagraph*{#1}}
\newcommand\paragraphstar[1]{\paragraph*{#1}}

\usepackage{authblk}

\title{A Parameterized Complexity Analysis of Bounded Height Depth-first Search Trees}

\author[1]{Lars Jaffke}
\author[2]{Paloma T.\ de Lima}
\author[3]{Wojciech Nadara}
\author[4]{Emmanuel Sam}

\affil[1]{NHH Norwegian School of Economics, Bergen, Norway}
\affil[2]{IT University of Copenhagen, Copenhagen, Denmark}
\affil[3]{University of Warsaw, Warsaw, Poland and Technical University of Denmark, Lyngby, Denmark}
\affil[4]{University of Bergen, Bergen, Norway}

\date{\today}

\begin{document}

\maketitle

\begin{abstract}
    Computing bounded depth decompositions is a bottleneck in many applications of the treedepth parameter. 
    The fastest known algorithm,
    which is
    due to Reidl, Rossmanith, S{\'{a}}nchez Villaamil, and Sikdar [ICALP 2014], 
    runs in $2^{\Oh(k^2)}\cdot n$ time
    and it is a big open problem whether the dependency on $k$ can be improved to $2^{o(k^2)}\cdot n^{\Oh(1)}$.
    We show that the related problem of finding DFS trees of bounded height can be solved faster in $2^{\Oh(k \log k)}\cdot n$ time.
    As DFS trees are treedepth decompositions,
    this circumvents the above mentioned bottleneck for this subclass of graphs of bounded treedepth.
    This problem has recently found attention independently 
    under the name 
    \MinHProb
    and our algorithm gives a positive answer to an open problem posed by Golovach [Dagstuhl Reports, 2023].
    We complement our main result by studying the complexity of \MinHProb and related problems in several other settings.
    First, we show that it remains NP-complete on chordal graphs, 
    and give an FPT-algorithm on chordal graphs for the dual problem, asking for a DFS tree of height at most $n-k$, parameterized by $k$.
    The parameterized complexity of \DualMinHProb on general graphs is wide open.
    Lastly, we show that \DualMinHProb and two other problems concerned with finding DFS trees with few or many leaves are FPT parameterized by $k$ plus the treewidth of the input graph.
\end{abstract}

\section{Introduction}
Treedepth~\cite{Nesetril2012} is a graph parameter with many interesting algorithmic applications.
As it is more restrictive than treewidth, 
there are several problems that can be solved more efficiently on graphs of bounded treedepth than on graphs of bounded treewidth,
in particular in terms of space complexity~\cite{FurerY17,HegerfeldK20,NederlofPSW23,PilipczukS21,PilipczukW18},
and in some cases in terms of time complexity~\cite{DvorakK18,KellerhalsK22}.
One downside of treedepth is that the current fastest algorithm to compute decompositions of depth $k$ runs in $2^{\Oh(k^2)}\cdot n$ time~\cite{Reidl14},
and it is a big open problem to determine whether the dependence on $k$ can be improved to $2^{o(k^2)}$ --- even if we allow a higher degree of the polynomial dependence on $n$,
or if we only want to have a constant-factor approximation instead of an optimal decomposition.

Bounded treedepth decompositions are rooted trees $T$ of bounded height 
on the vertex set of a graph $G$ that do not give rise to any cross-edges, 
i.e., for each edge $uv \in E(G)$, $u$ and $v$ are in an ancestral relationship in $T$. 
This property is shared with trees resulting from executing a depth-first search (DFS) on $G$;
with the difference that in a treedepth decomposition $T$, the edges of $T$ need not be edges of $G$.
Therefore, the class of graphs that have a DFS tree of bounded height is a subclass of the class of graphs of bounded treedepth.
Our first main result is that in this subclass of graphs of bounded treedepth, 
decompositions of depth $k$ can be computed faster, 
namely in $2^{\Oh(k \log k)}\cdot n$ time, see below.

The problem of computing DFS trees of small height has recently been studied 
by Sam, Fellows, Rosamond, and Golovach~\cite{SamFRG23},
motivated by several algorithmic applications~\cite{BODLAENDER19931,downeyFellow2013,fellows1989Polytime,freuder1985}.
They referred to DFS trees as \emph{lineal topologies},
and we keep the problem names introduced in their work for consistency.
\fancyproblemdef
    {Minimum Height Lineal Topology (\MinHProb)}
    {Graph $G$, integer $k$}
    {Does $G$ have a DFS tree (lineal topology) of height at most $k$?}

Fellows, Friesen, and Langston~\cite{fellows1988} showed that the the above problem is NP-complete
and hard to approximate.
Sam et al.~\cite{SamFRG23} observed that \MinHProb parameterized by the target height $k$ is FPT via the following win/win-approach.
If a graph has a DFS tree of height at most $k$, then it has treedepth at most $k$, which in turn implies that the treewidth is at most $k$.
The property of having a DFS tree with height at most $k$ can be expressed in MSO which yields the result via Courcelle's Theorem for bounded treewidth graphs~\cite{Courcelle90}.
However, an explicit algorithm for \MinHProb was not known, and asked for as an open problem by Golovach~\cite{DagstuhlDrawing}.
Our first contribution is a resolution of this problem via the following explicit algorithm.
\begin{restatable}{theorem}{thmMinHProb}\label{thm:fpt-td2}
    There is an algorithm that solves \MinHProb 
    given an $n$-vertex graph $G$ and an integer $k$ 
    in $2^{\Oh(k \log k)} \cdot n$ time.
    If $(G, k)$ is a yes-instance, 
    then the algorithm returns a certifying DFS tree.
\end{restatable}

In~\cite{SamFRG23}, several variations of \MinHProb are studied from the classical and parameterized points of view, and the coarse (parameterized) complexity of most problems is understood, apart from one notable exception: 
\fancyparproblemdef
    {Dual Minimum Height Lineal Topology (\DualMinHProb)}
    {Graph $G$ on $n$ vertices, integer $k$}
    {$k$}
    {Does $G$ have a DFS tree (lineal topology) of height at most $n-k$?}

Curiously, all attempts at classifying the parameterized complexity of 
\DualMinHProb have remained fruitless. 
We conjecture that it is $W[1]$-hard, however there no matching XP algorithm 
for that claim and to the best of our knowledge, 
it is not known whether it is solvable in polynomial time even for the case $k=2$. 
While we do not resolve this problem, we give insight into it by considering its parameterized complexity on \emph{chordal} graphs by giving an FPT-algorithm for it.
To justify the relevance of this algorithm, we first show that the non-parameterized version of \MinHProb, and hence of \DualMinHProb, remains NP-complete on chordal graphs.
\begin{theorem}
    \MinHProb (and hence \DualMinHProb) remains NP-complete on chordal graphs. 
    Moreover, \DualMinHProb can be solved in $\Oh(n^2) + 2^{\Oh(k^2)}$ time on chordal graphs.
\end{theorem}

Lastly, we consider the parameterized complexity of \DualMinHProb on 
graphs of bounded treewidth.
We observe that this property is expressible in MSO$_2$ and therefore FPT by $k$ plus the treewidth of the input graph.
As by-products, we observe the same for 
\MinLLTProb and \MaxLLTProb,
which ask for DFS trees with at most and at least $k$ leaves, respectively.
\begin{theorem}\label{thm:mso}
    \DualMinHProb, \MinLLTProb, and \MaxLLTProb parameterized by $k$ plus the treewidth of the input graph are fixed-parameter tractable.
\end{theorem}

\subparagraphstar{Related work.}
On the complexity of finding DFS trees with many or few leaves, Sam et al.~\cite{SamFRG23} proved that the \textsc{MinLLT} and \textsc{MaxLLT} problems are NP-complete. When parameterized by $k$, they showed that \textsc{MinLLT} is para-NP-hard and \textsc{MaxLLT} is W[1]-hard. By an application of Courcelle's theorem, they demonstrated that the parametric duals of these problems, namely, \textsc{DualMinLLT} and \textsc{DualMaxLLT}, respectively, are both FPT parameterized by $k$. Subsequently, it was shown in \cite{SamBPN23} that both problems admit polynomial kernels with $\Oh(k^3)$ vertices, which implies FPT algorithms running in  $k^{\Oh(k)}\cdot n^{O(1)}$ time. The authors also gave polynomial kernels for \textsc{Min-LLT} and \textsc{Max-LLT} parameterized by the vertex cover number. 

There has been extensive study on the computation of the treedepth of graphs~\cite{HansBod1998,HansBodGHT1995,FominGS2013,kobayashiTamaki2016}.
Regarding the complexity of \textsc{Treedepth} parameterized by $k$, Ossona de Mendez and Nešetřil showed by Courcelle's theorem that the problem is FPT parameterized by $k$. Later, Reidl, Rossmanith, Sanchez Villaamil, and Sikdar \cite{Reidl14} gave an algorithm that solves this problem explicitly in time $2^{\Oh(k\cdot\mathsf{tw}(G))} \cdot n$ and $2^{\Oh(k^2)} \cdot n$ via dynamic programming on tree decompositions (which can be $2$-approximated in $2^{\Oh(k)}\cdot n$ time~\cite{Korhonen21}). This algorithm improves on a previous algorithm by Bodlaender~\cite{HansBod1998} that runs time in $2^{\Oh(k\cdot\mathsf{tw}(G)^2)} \cdot n^2$ time.
Recently Nadara et al. \cite{Nadara22} presented an algorithm that computes treedepth in $2^{\Oh(k^2)} \cdot n^{\Oh(1)}$ time and $n^{\Oh(1)}$ space, improving on the algorithm by Reidl et al., which uses exponential space.

\subparagraphstar{Overview and methods.}
In Section \ref{sectn:Np-hard}, we show that \MinHProb remains NP-hard on chordal graphs via a polynomial-time reduction from the strongly NP-hard \textsc{Weighted Treedepth} problem~\cite{ipl2006} to \MinHProb on chordal graphs.

In Section \ref{sectn:DualMinhlt} we present an algorithm that, exploiting the structure provided by the clique tree of the chordal graph checks whether $\MinHFunc(G) \le n-k$ in $\Oh(n^2) + 2^{\Oh(k^2)}$ time and $\Oh(n^2)$ space. 

In \Cref{sectn:faster-MinLT-fpt-algo} we present an algorithm that, given a graph $G$ and a positive integer $k$, checks if $\MinHFunc(G) \le k$ in $k^{\Oh(k)} \cdot n$ time.  It relies on the fact that if $G$ is a YES-instance, then $\mathsf{tw}(G) \le \mathsf{td}(G) \le k$. It uses a tree decomposition of width $2k$ obtained by the $2$-approximation algorithm described in~\cite{Korhonen21} and adapts the strategy for computing $\mathsf{td}(G)$ by Reidl, Rossmanith, Sanchez Villaamil, and Sikdar \cite{Reidl14}. However, notably, we are able to improve over the natural adaptation of that algorithm that would lead to a $2^{\Oh(k^2)}$ time complexity.

In Section \ref{sectn:courcelle-dualhlt}
we show that the \DualMinHProb problem parameterized by $k$ plus treewidth is FPT on general graphs. The proof is based on applying Courcelle's theorem  \cite{Courcelle90}. Thus, we express the property of $\mathsf{MinHLT}(G) \le n-k$ in \textit{monadic second-order logic}. As a by-product, we obtain that the \textsc{MinLLT} and \textsc{MaxLLT} problems initiated by Sam et al.~\cite{SamFRG23} are also FPT parameterized by $k$ plus treewidth.

\section{Preliminaries}

We consider simple, finite, undirected, and connected graphs. For a graph $G$ with the vertex set $V(G)$ and the edge set $E(G)$, we use $n$ and $m$ to denote the number of vertices $|V(G)|$ and the number of edges $|E(G)|$ of $G$, respectively, if this does not cause any confusion. 

A graph $G$ is a \emph{subgraph} of $H$,
written $G \subseteq H$, 
if $V(G) \subseteq V(H)$ and $E(G) \subseteq E(H)$.
A graph $G$ is an \emph{induced subgraph} of $H$ 
if $V(G) = S \subseteq V(H)$ 
and $E(G) = \{uv \mid uv \in E(G), \{u, v\} \subseteq S\}$.
In this case we say that $G$ is the \emph{subgraph of $H$ induced by $S$}, written $G = H[S]$.

Let $G$ be a graph.
A sequence $v_1, \ldots, v_r$ of vertices of $G$ is a \emph{walk from $v_1$ to $v_r$} 
if for all $i \in [r-1]$, $v_i v_{i+1} \in E(G)$.
We say that $v_1, \ldots, v_r$ is a \emph{path} if all vertices are pairwise distinct,
and a \emph{cycle} if $v_1, \ldots, v_{r-1}$ are pairwise distinct and $v_r = v_1$.
A graph $G$ is \emph{connected} if there is a path from any vertex of $G$ to any other vertex.
A \emph{(connected) component} of $G$ is a maximal connected subgraph.
We may denote connected components as vertex sets when no confusion is possible.

A graph is a \emph{forest} if it does not contain a cycle and a connected forest is a \emph{tree}.
A tree $T$ is a \emph{spanning tree of $G$} if $T \subseteq V(G)$ and $V(T) = V(G)$.
A tree $T$ is \emph{rooted} if it has a unique designated vertex $r \in V(T)$, called the \emph{root}, 
inducing an ancestral relationship on some pairs of vertices in $T$:
We say that $u$ is an \emph{ancestor} of $v$ if $u$ lies on the path from $r$ to $v$ in $T$.
In this case we also call $v$ a \emph{descendant} of $u$.
If either $u$ is an ancestor of $v$ or $v$ is an ancestor of $u$, we say that $u$ and $v$ are in an \emph{ancestral relationship}, or \emph{comparable in $T$}.
A path $P$ in $T$ is \emph{vertical} if all pairs of vertices on it are comparable.
If $u$ is an ancestor of $v$ and $uv \in E(T)$,
then $u$ is the \emph{parent} of $v$ and $v$ the \emph{child} of $u$.
A \emph{rooted subtree} of a rooted tree $T$ consists of a vertex $v \in V(T)$ along with all of its descendants.
A forest $F$ where all components are rooted trees is called a \emph{rooted forest}.
We denote the set of roots of $F$ by $\roots(F)$.
(With slight abuse of notation, we may treat $F$ as a regular, unrooted, forest whenever convenient.)

Let $T$ be a rooted tree.
If $v$ is a leaf, then its \emph{height} is~$1$.
Otherwise, the \emph{height of $v$} is the maximum height of any of its children plus~$1$.
The \emph{height of a rooted tree} is the height of its root.
The \emph{depth} of the root is $1$, 
and the depth of any non-root node is the depth of its parent plus~$1$.

\paragraphstar{DFS trees.}
Let $G$ be a connected graph and $T \subseteq G$ a spanning tree of $G$.
We say that an edge $uv \in E(G)$ is a \emph{cross-edge for $T$} if $u$ and $v$ are not comparable in $T$.
We say that $T$ is a \emph{DFS tree of $G$} if $G$ has no cross-edges for $T$.

\paragraphstar{Tree Decompositions.}
\begin{definition}\label{def:td}
    Let $G$ be a graph.
    A \emph{tree decomposition} of $G$ is a pair $(T, \bags)$ 
    where $T$ is a tree and $\bags \colon V(T) \to 2^{V(G)}$ is such that:
    \begin{enumerate}
        \item $\bigcup_{t \in V(T)} \bags(t) = V(G)$.
        \item For all $uv \in E(G)$, there is some $t \in V(T)$ such that $\{u, v\} \subseteq \bags(t)$.
        \item For all $v \in V(G)$: $T[\{t \mid v \in \bags(t)\}]$ is connected.
    \end{enumerate}
    For each $t \in V(T)$, we call the set $\bags(t)$ the \emph{bag at $t$}.
    The \emph{width} of a tree decomposition $(T, \bags)$ is $\max_{t \in V(T)} |\bags(t)| - 1$.
\end{definition}

\begin{definition}\label{def:nice:td}
    Let $G$ be a graph and $(T, \bags)$ be a tree decomposition of $G$.
    We say that $G$ is a \emph{nice tree decomposition} 
    if $T$ is a rooted binary tree and each node $t \in V(T)$ is one of the following types.
    \begin{description}
        \item[Introduce Vertex.] The node $t$ either has one child $t'$ and $\bags(t) = \bags(t') \cup \{v\}$ for some vertex $v \notin \bags(t')$ or $t$ is a leaf and $\bags(t) = \{v\}$.
        We say that $v$ is \emph{introduced at $t$}.
        \item[Forget Vertex.] The node $t$ has one child $t'$ and $\bags(t') = \bags(t) \cup \{v\}$ for some $v \notin \bags(t)$.
        We say that $v$ and all edges incident to it are \emph{forgotten at $t$}.
        \item[Join.] The node $t$ has two children $t_1$ and $t_2$ with $\bags(t) = \bags(t_1) = \bags(t_2)$.
    \end{description}
    For each $t \in V(T)$, we denote by $G_t$ the subgraph of $G$ induced on all vertices that were introduced at $t$ or below.
\end{definition}

In a nice tree decomposition we may assume that each vertex is forgotten exactly once (but may be introduced multiple times), and that the number of nodes is $w^{\Oh(1)}n$, where~$w$ is its width and~$n$ the number of vertices of the underlying graph~\cite{CyganEtAl15}.

\paragraphstar{Treedepth decompositions.}

\begin{definition}
Let $G$ be a graph. A treedepth decomposition of $G$ is a rooted forest $F$ such that $V(F) = V(G)$ and for every edge $uv \in E(G)$, we have that either $u$ is the ancestor of $v$ in $F$ or $v$ is an ancestor of $u$.
\end{definition}

We note that the definition of a treedepth decomposition is very similar to the definition of DFS tree with one notable difference --- DFS trees have to additionally satisfy that edges of $F$ are also edges of $G$. The treedepth of $G$, denoted by $\td(G)$, is the lowest height a treedepth decomposition of $G$ might have. As every DFS tree of $G$ is also a treedepth decomposition of $G$, we get that $\td(G) \le \MinHFunc(G)$. It is known that $\MinHFunc(G) \le 2^{\td(G)}$, as treedepth of a path on $2^k$ vertices equals $k+1$ \cite{Nesetril2012} and $\MinHFunc(G)$ is at most the size of the longest path in $G$, so these two parameters are functionally equivalent. 

\paragraphstar{Chordal graphs.}
A graph $G$ is \emph{chordal} if it has no cycle of length four or more as an induced subgraph.
It is well-known that $G$ is chordal if and only if it has a tree decomposition where all bags are maximal cliques.

\section{\MinHProb is NP-hard on chordal graphs} \label{sectn:Np-hard}

In this section we will show the following theorem:

\begin{theorem} \label{thm:NP-hard-chordal}
The \MinHProb problem is NP-hard on chordal graphs. 
\end{theorem}

Subsequently, we deduce that the \DualMinHProb is also NP-hard on chordal graphs.

We will reduce from the problem known as \ProblemName{Weighted Treedepth} on trees. In the \ProblemName{Weighted Treedepth} problem we are given a graph with a positive weight $w(v)$ on each vertex $v$ and an integer $k$ and we ask whether there exists a treedepth decomposition of $G$ such that on any vertical path, the sum of the weights of the vertices on it is at most $k$. 
The lowest value of $k$ for which such treedepth decomposition exists will be denoted by $\td(G, w)$.

Dereniowski and Nadolski proved the following:
\begin{theorem}[Dereniowski and Nadolski~\cite{ipl2006}]\label{thm:wtdisHard}
    {\sc Weighted Treedepth} with positive integral weights is strongly NP-hard on trees. 
\end{theorem}

Given an instance $(T,w)$ of \ProblemName{Weighted Treedepth}, where $T$ is a tree, and a positive integer $M$, we construct a blown-up version of $T$ --- the graph $G_M(T,w)$ --- as follows. For each $x\in V(T)$ we create a set $V_x$ of size $w(x)\cdot M$ that induces a clique in $G_M(T,w)$. For each edge $xy\in E(T)$ we add all possible edges between the vertices of $V_x$ and $V_y$ in $G_M(T,w)$. A similar construction (with $M=1$) has been used in~\cite{ipl2006}, where the authors show the obtained graph is indeed chordal. 
They also proved the following.
\begin{lemma}[Dereniowski and Nadolski~\cite{ipl2006}]\label{lemma:tdTw}
$\td(G_1(T,w))=\td(T,w)$
\end{lemma}

We show the following, for a positive integer $M$.
\begin{observation}\label{obs:tdMTw}
$\td(G_M(T,w))=M\cdot\td(T,w)$
\end{observation}
\begin{proof}
Let $w_M$ be the function such that for every $v\in V(T)$, $w_M(v)=M\cdot w(v)$. By Lemma~\ref{lemma:tdTw}, $\td(G_1(T,w_M))=\td(T,w_M)$. Note that for any tree $T$, $\td(T,w_M)=M\cdot \td(T,w)$. Furthermore, note also that by construction, $G_1(T,w_M)$ is isomorphic to $G_M(T,w)$. Hence the observation follows.
\end{proof}

We base our proof on the following lemma:
\begin{lemma}\label{lem:sandwich}
    Let $(T,w)$ be an instance of {\sc Weighted Treedepth} where $T$ is a tree with $n$ vertices and let $M\in \mathbb{Z}^+$ be such that $M\geq n$. Then 
    $$M\cdot\td(T,w)\leq \MinHFunc(G_M(T,w)) \leq M\cdot \td(T,w) + n^2.$$
\end{lemma}
\begin{proof}
By Observation \ref{obs:tdMTw}, 
$M\cdot\td(T,w)= \td(G_M(T,w))$. Furthermore, since any DFS tree is also a treedepth decomposition tree, we get that $\td(G_M(T,w))\leq \MinHFunc(G_M(T,w))$. Hence, the first inequality follows.

To prove the second inequality, we proceed as follows. We will define a recursive procedure to construct a DFS tree for $G_M(T,w)$ (shortened as $G$) with height at most $M\cdot \td(T,w) + n^2$. 

Let $D$ be an optimal treedepth decomposition for $(T,w)$. Without loss of generality, we may assume that $D$ is \emph{recursively optimal}, which means that any rooted subtree of $D$ is an optimal treedepth decomposition for the subgraph of $T$ induced on its vertices. 
(This in particular implies that any rooted subtree of $D$ induces a connected subgraph of $T$ \cite{DynamicTreedepth}.) 
Let $r\in V(T)$ be the root of $D$. From now on, we consider the tree $T$ also to be rooted at $r$. We now define the following procedure $\const(S,w_S,v)$, where $S$ is some subgraph of $T$, $w_S$ is a weight function on its vertices and $v$ is a vertex of $S$. This function is meant to output a DFS tree of $G_1(S, w_S)$ whose height is at most $\td(S, w_S) + |S|^2$.
We additionally maintain two invariants that will be satisfied throughout all of our calls of that function. The first invariant is that for any $u \in V(S)$ we have that $w_S(u) \ge |V(S)|$. The second invariant is that $V(S)$
is the set of vertices of the subtree of $D$ rooted at $v$. These invariants are satisfied for parameters $(S, w_S, v) \coloneqq (T, w_M, r)$, as $w_M(u) = M \cdot w(u) \ge n=|V(T)|$ and the subtree of $D$ rooted at $r$, which is its root, is simply $D$ itself and $V(D) = V(T)$.

The procedure call $\const(S, w_S, v)$ proceeds as follows:
\begin{enumerate}
	\item The DFS visits all the vertices $V_v\subseteq V(G)$ in an arbitrary order (we remind that $V_v$ is the clique on $w_S(v)$ vertices corresponding to $v$).

    Note that the connected components of $G_1(S, w_S) \setminus V_v$ correspond to graphs $G_1(C_1, w_{C_1}),$ $\ldots, G_1(C_c, w_{C_c})$, where $C_1, \ldots, C_c$ are the connected components of $S \setminus \{v\}$ and $w_{C_i} \coloneqq w_S|_{C_i}$, so the DFS tree has to branch into each $C_i$ separately.
    Note that $v$ has exactly one child in $D$ per each $C_i$. 
	\item For each $i=1, \ldots, c$, let us do the following. Let $x_i$ be the child of $v$ in $D$ that belongs to $C_i$ and let $P_i$ be a path between $v$ and $x_i$ whose all vertices except for $v$ are in $C_i$. Suppose $P_i=vy_1y_2\ldots y_{p_i}$, with $y_{p_i}=x_i$. For each $1\leq j \le p_i-1$, the DFS visits an arbitrary vertex $z_j\in V_{y_j}$, that is, we add edges $z_0z_1$, $z_1z_2$, \ldots, $z_{p_i-2}z_{p_i-1}$ to the constructed tree, where $z_0$ is the last visited vertex of $V_v$.
	 
    Note that we are allowed to do that as $w_S(y_j) \ge |S| \ge 1$, so $V_{y_j}$ is nonempty. Let $G'_i$ be the graph obtained from $G_1(C_i, w_{C_i})$ by removing, for each $1\leq j\leq p_i - 1$, the vertex $z_j$ from $V_{y_j}$. Let $w'_{C_i}$ be $w_{C_i}$ modified in a way reflecting visiting the vertices $z_1, \ldots, z_{p_i-1}$, that is, 
    $w'_{C_i}(u) = w_{C_i}(u)-1$ 
    for $u \in \{y_1, \ldots, y_{p_i-1}\}$ and $w'_{C_i}(u) = w_{C_i}(u)$ otherwise.
    \item Call $\const(C_i,w'_{C_i},x_i)$ and append the resulting DFS tree as a child subtree of $z_{p-1}$.
 
    Note that our invariants are satisfied for this call as we have $w'_{C_i}(u) \ge w_{C_i}(u) -1 = w_{S}(u) -1 \ge |S|-1 \ge |C_i|$ and $C_i$ is indeed the set of vertices of the subtree of $D$ rooted at $x_i$. Also note that the root of the tree returned by that call belongs to $V_{x_i}$, so it is adjacent to $z_{p_i-1}$, so all edges of the final tree will indeed be edges of $G$.
\end{enumerate}

It is clear from the construction that the tree $F$ produced by the call $\const(T, w_M, r)$ satisfies that $G_M(T, w)$ does not have any cross-edges for $F$. It is also clear that $E(F) \subseteq E(G)$ and $V(F) = V(G)$, hence $F$ is indeed a DFS tree of $G_M(T, w)$. 

What is left to prove is the claimed upper bound on the height of $F$. We do so by the induction on $|S|$, where we claim that $\const(S, w_S, v)$ outputs a DFS tree of $G_1(S, w_S)$ of height at most $\td(S, w_S) + |S|^2$. That statement is obvious for $|S|=1$. For the inductive step we note that since $D$ is recursively optimal and $D[V(S)]$ is a subtree of $D$ rooted at $v$, we have that $\td(S, w_S) = w_S(v) + \max(\td(C_1, w_{C_1}), \ldots, \td(C_c, w_{C_c}))$. 
The final tree produced by the $\const(S, w_S, v)$ call consists of a path on $|V_v| = w_S(v)$ vertices with $c$ different branches attached to its deepest vertex, where the $i$-th branch has height at most $p_i-1 + \td(C_i, w'_{C_i}) + |C_i|^2$ thanks to our inductive hypothesis. However, note that  $\td(C_i, w'_{C_i}) \le \td(C_i, w_{C_i})$ and $p_i-1 + |C_i|^2 \le |S|-2 + (|S|-1)^2 \le |S|^2$, so the height of the returned tree is at most $w_S(v) + \max(\td(C_1, w_{C_1}), \ldots, \td(C_c, w_{C_c})) + |S|^2 = \td(S, w_S) + |S|^2$, as claimed.
\end{proof}

Having proven \Cref{lem:sandwich}, we are ready to prove \Cref{thm:NP-hard-chordal}.

\begin{proof}[Proof of \Cref{thm:NP-hard-chordal}]
    We give a reduction from {\sc Weighted Treedepth} with positive integral weights on trees. By Theorem~\ref{thm:wtdisHard}, this problem is NP-hard even if the weights are polynomial in the size of the tree. 
    Let $(T,w)$ be an instance of {\sc Weighted Treedepth} with positive integral weights that are polynomial in $|V(T)|=n$. Let $M=n^2+1$. As $M$ is polynomial in $n$, the function $w_M(v) = M \cdot w(v)$ stays polynomial in $n$. Consequently, the size of $G_M(T, w)$ is polynomial in $n$ as well. Moreover, $G_M(T, w)$ is a chordal graph~\cite{ipl2006}. 
    Based on \Cref{lem:sandwich}, we have that 
    $$M\cdot\td(T,w)\leq \MinHFunc(G_M(T,w)) \leq M\cdot \td(T,w) + n^2 < M \cdot (\td(T, w) + 1).$$ 
    Therefore, we can conclude that $\td(T, w) = \floor{\frac{\MinHFunc(G_M(T, w))}{M}}$. Hence, we deduce that $\MinHProb$ is NP-hard on chordal graphs.
\end{proof}

\section{\DualMinHProb is FPT on chordal graphs} \label{sectn:DualMinhlt}

This section will prove that \DualMinHProb is FPT on chordal graphs. More formally, we prove the following theorem:
\begin{theorem}\label{thm:fpt-chordal}
Let $G$ be a chordal graph on $n$ vertices and $m$ edges and $k$ be a positive integer. Then, we can check if $\MinHFunc(G) \le n-k$ in $\Oh(n^2) + 2^{\Oh(k^2)}$ time and $\Oh(n^2)$ space. Moreover, if this inequality holds, then we can return a certifying DFS tree in the same time and space complexities.
\end{theorem}

In preparation for proving this theorem, we will first prove some auxiliary lemmas.

\begin{lemma}\label{lem:ham-sep}
Let $G$ be a graph, $S \subseteq V(G)$ be a subset of vertices such that $G[S]$ has a Hamiltonian path and $C_1, C_2, \ldots, C_c \subseteq V(G)-S$ be a partition of $V(G)-S$ into connected components. 
Then $\MinHFunc(G) \le |S| + \max(|C_1|, \ldots, |C_c|)$. Moreover, if we are given any Hamiltonian path of $G[S]$, then we can return a certifying DFS tree in $\Oh(|E(G)|)$ time. 
\end{lemma}
\begin{proof}
Let $T$ be any DFS tree resulting from starting a DFS traversal of our graph by traversing the provided Hamiltonian path $P$ within $G[S]$. It can be easily seen that if $v \in C_i$ and $u \in C_j$ for $i \neq j$ then $u$ and $v$ will not be comparable in $T$. 
Hence, any vertical path can contain vertices from at most one of the sets $C_1, \ldots, C_c$. Therefore, it cannot contain more than $\max(|C_1|, \ldots, |C_c|)$ vertices from $C_1 \cup \ldots \cup C_c$.
As a consequence, every vertical path in $T$ is of height at most $|S| + \max(|C_1|, \ldots, |C_c|)$, hence $\MinHFunc(G) \le |S| + \max(|C_1|, \ldots, |C_c|)$. Also, if we are given $P$, then we can compute $T$ in $\Oh(|E(G)|)$ time.
\end{proof}

We need the following lemma which has been proved as Lemma 2.5 in \cite{GraphMinorsII}.
We would like to point out that the statement of \cite[Lemma 2.5]{GraphMinorsII} is weaker than the following in two ways.
First, its statement is purely existential, while we need to be able to compute the desired set in polynomial time.
Second, it only speaks about a set of vertices, 
while we require the set to be a subset of a bag of a tree decomposition.
However, a close inspection of the proof of \cite[Lemma 2.5]{GraphMinorsII}
reveals that it proves the following, stronger, statement.
\begin{lemma}[Cf.~Lemma 2.5 in \cite{GraphMinorsII}]\label{lem:balanced-sep-tw}
    Let $G$ be an undirected $n$-vertex graph and $(T, \beta)$ be its tree decomposition of width $w$. 
    There exists $K \subseteq V(G)$, which is a subset of some bag $\beta(v)$, such that if $C_1, \ldots, C_c$ are vertex sets of connected components of $G-K$, then we have $2 \cdot \max(|C_1|, \ldots, |C_c|) \le |C_1| + |C_2| + \ldots + |C_c|$. Moreover, $K$ can be found in $\Oh(n\cdot w)$ time.
\end{lemma}

We proceed to the most challenging part of proving \Cref{thm:fpt-chordal}, which is the following lemma:

\begin{lemma}\label{lem:moderate-case}
Let $G$ be a chordal graph, $k$ be a positive integer and $K$ be a clique in $G$ such that $n-2k<|K| \le n-k$ given to us. Then, we can check if $\MinHFunc(G) \le n-k$ in $\Oh(n+m) + 2^{\Oh(k^2)}$ time and $\Oh(n + m)$ space.  Moreover, if this inequality holds, then we can return a certifying DFS tree in the same time and space complexities.
\end{lemma}

This lemma is a restriction of \Cref{thm:fpt-chordal} to the hardest case of a ``moderate clique''. 
We will firstly see how to prove \Cref{thm:fpt-chordal} assuming correctness of \Cref{lem:moderate-case}, and then we will focus on proving \Cref{lem:moderate-case}.

\begin{proof}[Proof of \Cref{thm:fpt-chordal}]
We first compute a clique tree $\Tt$ of $G$, that is, a tree decomposition where each bag is a maximal clique, in linear time~\cite{LexBFS}.
If $G$ is a clique, then we obviously have $\MinHFunc(G)=n$. If $G$ is not a clique, then we can apply \Cref{lem:balanced-sep-tw} and get a set $K$ such that we have that $2 \cdot \max(|C_1|, \ldots, |C_c|) \le |C_1| + |C_2| + \ldots + |C_c|$, where $C_1, \ldots, C_c$ are vertex sets of connected components of $G-K$. 
Since $K$ is the subset of some bag of $\Tt$, it is a clique.
Let us consider a few cases based on the size of $K$:
\begin{description}
    \item[The ``small clique case'':] $|K| \le n-2k$
    \item[The ``moderate clique case'':] $n - 2k < |K| \le n - k$
    \item[The ``large clique case'':] $n-k<|K|$
\end{description}

In the large clique case we clearly have $\MinHFunc(G) \ge |K| > n-k$, so we return $NO$.
In the small clique case we apply \Cref{lem:ham-sep} (as $K$ is clearly Hamiltonian and we can easily get its Hamiltonian path as well) and get that $\MinHFunc(G) \le |K| + \max(|C_1|, \ldots, |C_c|)$ and get a certifying DFS tree as well. Note that $|K| + \max(|C_1|, \ldots, |C_c|) \le \frac{2|K| + |C_1| + \ldots + |C_c|}{2} = \frac{n + |K|}{2} \le n-k$, hence we can return $YES$ and a certifying DFS tree as well.
What remains is the most challenging moderate clique case, which is handled by \Cref{lem:moderate-case}.
\end{proof}

We will now explain the proof of \Cref{lem:moderate-case}.

\begin{proof}[Proof of \Cref{lem:moderate-case}]
Let us denote $R \coloneqq V(G) - K$, that is, the set of vertices outside of $K$, and denote their number by $r$. Let us recall that in this case we have $k \le r < 2k$. 
Now, for $S \subseteq R$, 
let us define $g_S = |\{v | v \in K, N(v) \cap R = S\}|$, that is, the number of vertices of $K$ whose neighborhood in $R$ is $S$.

The main idea behind this algorithm is the fact that since $r$ is bounded from above by $2k$, the number of different neighborhoods in $R$ of vertices of $K$ is bounded by $2^{r} \le 2^{2k}$. As all vertices from $K$ are adjacent to each other, our instance can be in a sense compressed to an instance on $r+2^r$ vertices, where vertices have varying values of congestion (which means, that they should be used multiple times), where the $2^r$ part comes from compressing $g_S$ corresponding vertices of $K$ to one vertex with congestion $g_S$. The value $r+2^r$ is bounded from above as a function of $k$, but the quantities $g_S$ are unbounded in terms of $k$, so dealing with that issue is the main obstacle in showing that this problem is fixed parameter tractable when parameterized by $k$.

We will also define a \emph{type of a vertex $v$} as a function $\mathrm{type}(v) : V(G) \to (R \cup 2^R)$ as

\begin{equation*}
    \mathrm{type}(v) = 
        \begin{cases}
        v, & \text{for } v \in R \\
        N(v) \cap R & \text{for } v \in K
        \end{cases},
\end{equation*}
following the idea that vertices within one type are indistinguishable. In particular, knowing the types of two vertices is sufficient information to determine if they are adjacent. By a slight abuse of notation, we will say that two types are adjacent if and only if any two vertices belonging to said types are adjacent (which does not depend on the specific choice of vertices).

Let $T$ be any DFS tree of $G$. We will define its compressed version $\Compr(T)$ in the following way. Let us define the set $I \subseteq V(G)$ of \emph{important vertices} as the set consisting of $N_T[R]$, the root of $T$ and the deepest vertex of $K$ in $T$ (note that as all vertices of $K$ lie on a single vertical path in $T$ --- such vertex is well defined). Note that $N_T[V(G) \setminus I] \subseteq K$ and all vertices of $V(G) \setminus I$ have a parent (because the root is in $I$) and exactly one child (they cannot have zero children as the deepest vertex of $K$ belong to $I$ and they cannot have more than one, because all vertices of $K$ lie on a single vertical path). Hence, all vertices of $V(G) \setminus I$ have degree two in $T$, and the subgraph induced by all edges adjacent to at least one vertex of $V(G) \setminus I$ is a set of paths with ends in $K$. We get the compressed version $\Compr(T)$ of $T$ by bypassing all vertices from $V(G) \setminus I$, that is, replacing each such path with a single edge between its ends. 
Note that as all vertices of $K$ lie on a single vertical path of $T$, so every vertex of $R$ has at most two neighbors from $K$ in $T$, hence $|V(\Compr(T))| \le 3|R|+2 = \Oh(k)$. 
For any edge in $\Compr(T)$ both of whose ends belong to $K$, we will call it $\emph{amendable}$, denoting the fact that during the process of ,,decompressing'' $\Compr(T)$ to $T$, we may add some vertices of $K$ to this edge, turning it into a nontrivial path of vertices of $K$.

The plan is to iterate over all possible candidates $C$ for a compressed version of $T$, for each of them determine whether there exists any $T$ such that $\Compr(T) = C$, and if yes, then greedily determine the one minimizing its height. A candidate $C$ for $\Compr(T)$ comprises of a rooted tree and a mapping from vertices of $C$ to types of vertices from $V(G)$ (as vertices within one type are indistinguishable, we do not have to guess the full mapping from $V(C)$ to $V(G)$). All candidates $C$ have $\Oh(k)$ vertices and there are at most $2^{\Oh(k)}$ rooted trees on that many vertices and at most $(r+2^r)^{\Oh(k)} = 2^{\Oh(k^2)}$ mappings $V(C) \to (R \cup 2^R)$, hence there are at most $2^{\Oh(k^2)}$ candidates for $C$. 

Now, for a candidate $C$ with a mapping $\tau : V(C) \to (R \cup 2^R)$ and root $t$ we should run some sanity checks to see if it could be a compressed version of some DFS tree. First, for each $v \in R$ we should have $|\tau^{-1}(v)| = 1$ and for each $S \in 2^R$ we should have $|\tau^{-1}(S)| \le g_S$. For each edge $ab$ of $C$, the types $\tau(a)$ and $\tau(b)$ should be adjacent. Also, for each pair of incomparable vertices $a, b$ of $C$, the types $\tau(a)$ and $\tau(b)$ cannot be adjacent. 
All vertices mapped to types from $2^R$ have to lie on a single vertical path and let the deepest one be called $X$. 
Also, for $S \subseteq 2^R$, by $H(S)$, denote the set of vertices $v \in R$ such that $v$ and $S$ are adjacent and $\tau^{-1}(v)$ is not an ancestor of $X$ (at this point we have already performed the check that $|\tau^{-1}(v)|=1$, so by a slight abuse of notation we may treat $\tau^{-1}(v)$ as its only element rather than a set). The only nontrivial condition is the following:

\begin{lemma} \label{lem:amendable}
Let $S \in 2^R$ be such that $|\tau^{-1}(S)| < g_S$. Then, for any $v \in H(S)$ at least one of the edges on the path from $t$ to $\tau^{-1}(v)$
has to be amendable for $C$ to be a viable candidate (we recall that $t$ is the root of $C$).
\end{lemma}
\begin{proof}
This condition must hold because if $|\tau^{-1}(S)| < g_S$, then at least one vertex of type $S$ was lost during the compression process, so it must have been a part of a path that became an amendable edge. If $\tau^{-1}(v)$ is not an ancestor of $X$, then it is not an ancestor of any vertex from $\tau^{-1}(2^R)$, so all vertices from $\tau^{-1}(2^R)$ that are comparable to $\tau^{-1}(v)$ are ancestors of $\tau^{-1}(v)$ . Hence, if none of the edges on the path from $t$ to $\tau^{-1}(v)$ is amendable, then for any amendable edge, its deeper end is not comparable with $\tau^{-1}(v)$ and all vertices that it will give rise to when being decompressed will be incomparable to $v$ in $T$ (hence, non-adjacent to $v$ in $G$). Hence, no vertex that $v$ is adjacent to in $G$ might have been lost in the compression process, which contradicts the fact that $|\tau^{-1}(S)| < g_S$.
\end{proof}

Assume that $C$ fulfilled all of these conditions. We claim that there is at least one $T$ such that $\Compr(T) = C$ and we will determine one such tree minimizing the height. Note that the choice that we have while decompressing $C$ is to assign a mapping $L(e) : 2^R \to \mathbb{Z}_{\ge 0}$ for each amendable edge $e$, denoting the number of vertices of each type that were lost while compressing the path that became $e$ in $C$. Notice that as $K$ forms a clique, we automatically have the property that whatever order we put these new vertices in for every amendable edge, edges of $T$ will still be edges of $G$, so their order within the amendable edges paths does not matter. As all amendable edges lie on a single vertical path, if for a vertex type $S \in 2^R$, we consider all vertices $v \in H(S)$, then the respective sets of amendable edges on $vt$ paths are such that the intersection of all of them is equal to the smallest of them. Hence, we can distribute $g_S - |\tau^{-1}(S)|$ vertices of type $S$ over amendable edges that are within this smallest set in any way, and by doing so, we will ensure that no vertices of type $S$ are in a pair $vS$ violating the condition that adjacent vertices of $G$ have to be comparable in $T$. That concludes why these conditions are sufficient for the existence of any $T$ such that $T$ is a DFS tree of $G$ and $\Compr(T) = C$.

What is left is to determine the choice of $T$ minimizing its height. Let us observe that for each $S$ any distribution of $g_S - |\tau^{-1}(S)|$ vertices of type $S$ over the described set of amendable edges will lead to a viable $T$. As all amendable edges lie on a single vertical path, it is clear that the best choice is to put all $g_S - |\tau^{-1}(S)|$ of them into the deepest of these allowed amendable edges.

We keep the representation of $T$ concise using mappings $L(e)$ (as opposed to creating copies of all $n$ vertices per each viable $C$). The height of the DFS tree that would be an effect of decompression can easily be calculated. 
All computations regarding considering one $C$ can be performed in time $2^{\Oh(k)}$. If we find any feasible solution, we can then easily create a witnessing $T$ in an uncompressed form, that is, simply as a subgraph of $G$. Hence there are $2^{\Oh(k^2)}$ candidates $C$ and considering each of them takes $2^{\Oh(k)}$ time, so the whole algorithm works in $2^{\Oh(k^2)} + \Oh(n^2)$ time and $2^{\Oh(k)} + \Oh(n^2)$ space.
``Empty types'', i.e., types that do not have any vertex that match them, can be omitted in the space complexity analysis, which brings the space complexity down to $\Oh(n^2)$.
\end{proof}

\section{An explicit FPT-algorithm for \MinHProb} \label{sectn:faster-MinLT-fpt-algo}

In~\cite{SamFRG23},
Sam, Fellows, Rosamond, and Golovach showed that \MinHProb is FPT by solution size,
using Courcelle's Theorem.
In this section, we give an explicit algorithm for this problem 
which was mentioned as an open problem by Golovach in~\cite{DagstuhlDrawing}.

\thmMinHProb*

To prove \Cref{thm:fpt-td2}, 
we follow a similar strategy as the FPT-algorithm to compute treedepth due to 
Reidl, Rossmanith, Sanchez Villaamil, and Sikdar~\cite{Reidl14}. Perhaps the most interesting part about our algorithm is not the fact that it can be obtained as an adaptation of the algorithm from \cite{Reidl14}, but the fact that thanks to specific properties of DFS trees, the \MinHProb problem turns out to be solvable in $k^{\Oh(k)} \cdot n$ time, rather than $2^{\Oh(k^2)}$ time complexity which is required for the fastest known algorithms for computing treedepth exactly~\cite{Reidl14,Nadara22}, which would also be the time complexity of the natural adaptation of \cite{Reidl14} to \MinHProb.

For this section, instead of thinking that a DFS tree $F$ of a graph $G$ has the same set of vertices, we prefer thinking in a way that they have different vertex sets, but we are additionally provided with a bijective mapping $\rho$ from $V(G)$ to $V(F)$.

\begin{definition} \label{def:full-solution}
We say that $(F, \rho)$ is a \emph{full solution} for $G$ if $F$ is a tree and $\rho$ is a bijective mapping from $V(G)$ to $V(F)$ such that $uv \in E(G)$ implies that $\rho(u)$ and $\rho(v)$ are in an ancestral relation in $F$ (that is, one of them is the ancestor of the other in $F$) and $uv \in E(F)$ implies $\rho^{-1}(u) \rho^{-1}(v) \in E(G)$.
\end{definition}

The starting point is this:
If a graph has a DFS tree of height at most $k$,
then this tree witnesses that 
$G$ has treedepth at most $k$.
Since for each graph $G$, we have that
$\tw(G) \le \td(G)$,
and since 
we can compute a $2$-approximation of the treewidth of a graph $G$ in $2^{\Oh(\tw(G))}\cdot |V(G)|$ time~\cite{Korhonen21},
we may assume that we have a tree decomposition of width at most $2k$ of the input graph at hand.
Below, we describe a dynamic programming algorithm that uses this tree decomposition to decide if $G$ has a DFS tree of height at most $k$.
For convenience, we further assume that this tree decomposition is a \emph{nice tree decomposition}, see \Cref{def:nice:td}.

The following definition will become handy:
\begin{definition} \label{def:closure}
Let $T$ be a rooted tree and $X$ be a subset of its vertices. The ancestral closure of~$X$, denoted $\mathsf{Clos}_T(X)$, is the subgraph of $T$ induced on vertices of $X$ and all of their ancestors.
\end{definition}

Let us now define a \emph{partial solution for a graph $H$}.
\begin{definition} \label{def:partial-graph}
    A partial solution for a graph $H$ is a pair $(D, \tau)$, where $D$ is a rooted tree of height at most $k$ and $\tau : V(H) \to V(D)$ is an injective mapping such that for any $uv \in E(H)$ we have that $\tau(u)$ and $\tau(v)$ are in an ancestral relation in $D$, for any $uv \in E(D)$ such that $u, v \in \tau(V(H))$, we have that $\tau^{-1}(u)\tau^{-1}(v) \in E(H)$ and any leaf of $D$ belongs to $\tau(V(H))$.
\end{definition}

Intuitively, such partial solution represents a candidate leading to a full solution for a~supergraph $G$ of $H$, where the vertices of $D$ that are not in the $\tau(V(H))$ will correspond to some vertices of $V(G) \setminus V(H)$. Let us make this intuition formal through the following definition:

\begin{definition} \label{def:partial-full-extension}
Let $H$ be a subgraph of $G$, $(D, \tau)$ be a partial solution for $H$ and $(F, \rho)$ be a~full solution for $G$. We say that $(D, \tau)$ \emph{extends to} $(F, \rho)$ if there exists a mapping $\pi : V(D) \to V(F)$ that maps the root of $D$ to the root of $F$, respects edges and $\tau \circ \pi$ and $\rho$ agree on $V(H)$.
\end{definition}

Note that if $(D, \tau)$ extends to $(F, \rho)$, then $D = \mathsf{Clos}_F(\tau(V(H)))$. That holds because $D$ is a connected subgraph of $F$ containing the root of $F$ and the whole $\tau(V(H))$, hence $\mathsf{Clos}_F(\tau(V(H))) \subseteq D$, and if $V(D) \setminus V(\mathsf{Clos}_F(\tau(V(H)))$ was nonempty, then $D$ would contain a leaf that does not belong to $\tau(V(H))$, which contradicts the definition of a partial solution.

Let $\Tt = (T, \beta)$ be a tree decomposition of $G$ of width at most $2k$ that we have already computed. We will define \emph{partial solutions for nodes of $T$} as well.

\begin{definition} \label{def:partial-node}
A partial solution for a node $t \in V(T)$ is a triple $(D, \tau, S)$, where $D$ is a rooted tree, $\tau$ is an injective mapping from $\beta(t)$ to $V(D)$, and $S \subseteq V(D)$,
such that there exists a partial solution $(D', \tau')$ for $G_t$ such that 
\begin{itemize}
    \item $D = \mathsf{Clos}_{D'}(\tau(\beta(t)))$,
    \item $\tau$ and $\tau'$ agree on $\beta(t)$,
    \item $\tau'(V(G_t) \setminus \beta(t)) \cap V(D) = S$ and $\tau'(V(G_t) \setminus \beta(t)) \setminus V(D) = V(D') \setminus V(D)$.
\end{itemize}   
We additionally say that such $(D, \tau, S)$ \emph{extends to} $(D', \tau')$.
\end{definition}

Intuitively, a partial solution for a node $t$ is a partial solution for $G_t$ stripped off any redundant information relating to already forgotten vertices (where forgotten vertices means $V(G_t) \setminus \beta(t)$) that will never influence our future possibilities of extending such partial solution to a full solution for the whole graph $G$.

In this setting, the vertices $V(D')$ can be partitioned into the following parts:
\begin{enumerate}
    \item $\tau(\beta(t))$ --- images of the vertices from the node $t$
    \item $S$ --- vertices on paths from some $v \in \tau(\beta(t))$ to the root that have a vertex that is already forgotten as their preimage in $\tau'$
    \item $V(D) \setminus (\tau(\beta(t)) \cup S)$ --- vertices on paths from some $v \in \tau(\beta(t))$ to the root that might have a vertex that will be introduced later as their preimage in $\tau'$
    \item $V(D') \setminus V(D)$ --- vertices that are images in $\tau'$ of forgotten vertices and not on a path between $r$ and any $v \in \tau(\beta(t))$
\end{enumerate}
As vertices forgotten below $t$ are not adjacent to anybody from $V(G) \setminus V(G_t)$, we do not need to remember anything about how the part outside of $\mathsf{Clos}_{D'}(\tau(\beta(t)))$ looks like, and for any vertex $u \in V(D) \cap \tau'(V(G_t) \setminus \beta(t))$, we also do not need to remember who specifically was mapped to $u$ beside the fact that there existed a vertex like that. Consequently, we say that a~partial solution $(D, \tau, S)$ for $t$ extends to a full solution $(F, \rho)$ if there exists a partial solution $(D', \tau')$ for $G_t$ such that it extends to $(F, \rho)$ and $(D, \tau, S)$ extends to $(D', \tau')$. Similarly, if $(D, \tau, S)$ extends to $(F, \rho)$, we have that $D = \mathsf{Clos}_F(\tau(\beta(t)))$.

It is clear from the definitions that $\MinHFunc(G) \le k$ if and only if there exists a partial solution $(D, \tau, S)$ for the root $r$ of $T$ such that $V(D) = \tau(\beta(r)) \cup S$.

The crucial property differentiating treedepth decompositions and DFS trees is that the edges of DFS trees have to be edges of the original graph as well. That property has interesting implications for our partial solutions. Let us consider a partial solution $(D, \tau, S)$ for a node~$t$ and assume that there exist 
two vertices $u, v \in V(D)$
such that $u \in S$, $v \notin S \cup \tau(\beta(t))$ and $uv \in E(D)$. Such partial solution cannot be extended to a full solution for $G$ as in any extension of that partial solution to a full solution, $v$ has to be an image of a vertex from $V(G) \setminus V(G_t)$, while $u$ was an image of a vertex from $V(G_t) \setminus \beta(t)$. However, these two sets are not adjacent to each other in $G$, contradicting the fact that for an edge of a full solution, its preimages are adjacent to each other. Hence, we may restrict our computations only to partial solutions such that sets $S$ and $V(D) \setminus (S \cup \tau(\beta(t)))$ are not adjacent to each other as other partial solutions do not extend to any full solutions. Such partial solutions will be called \emph{promising}. In other words, each connected component of $D \setminus \tau(\beta(t))$ has to be either fully in~$S$, or disjoint from~$S$. We then prove the following:

\begin{lemma}\label{lem:promising-sets}
For a node $t$ and a pair $(D, \tau)$, there are only $2^{\Oh(k)}$ sets $S$ such that $(D, \tau, S)$ is a promising partial solution for $t$.
\end{lemma}
\begin{proof}
Assume that there exists $S$ such that $(D, \tau, S)$ is a promising partial solution for $t$, as otherwise the claim is vacuously true. Note that, as $D = \mathsf{Clos}_{D'}(\tau'(\beta(t)))$ for some partial solution $(D', \tau')$ for $G_t$ that it extends to, $D$ has at most one degree-one vertex that is not within $\tau(\beta(t))$ (the root of $D$ is the only possible vertex with such property). It implies that the number of connected components of $D \setminus \tau(\beta(t))$ is at most $|\beta(t)| \le 2k+1$. As each of them has to be either fully in $S$ or disjoint from $S$, there are at most $2^{2k+1} = 2^{\Oh(k)}$ possible values of $S$ such that $(D, \tau, S)$ is a promising partial solution.
\end{proof}

As the next step, we will prove the following claim:
\begin{lemma}\label{lem:trees-shapes}
For a node $t$, there are only $k^{\Oh(k)}$ non-isomorphic pairs $(D, \tau)$ across all partial solutions for $t$.
\end{lemma}
\begin{proof}
Let us first call the vertices of $D$ that are the images of $\tau$ \emph{marked}. For a rooted tree $D$ with a specific set of marked vertices $M$ of size $|\beta(t)|$, there are at most $|\beta(t)|! = k^{\Oh(k)}$ pairs $(D, \tau)$ which result in this specific pair $(D, M)$. In the following, we bound the number of non-isomorphic pairs $(D, M)$ (where the supposed isomorphisms have to map marked vertices to marked vertices and non-marked vertices to non-marked vertices).

Recall that $D$ is a union of at most $|\beta(t)| \le 2k+1$ paths of length at most $k$ from its root. Hence, it has at most $2k$ vertices of degree at least three and at most $2k+2$ vertices of degree one. Let us call all vertices of degree other than two \emph{special} and let the root and all marked vertices be special as well. All non-special vertices have degree two and form a set of paths between pairs of special vertices. If we contract these paths, we will get a rooted tree of size $\Oh(k)$ with marked vertices. There are only $2^{\Oh(k)}$ non-isomorphic rooted trees on $\Oh(k)$ vertices and for each of them it could have at most $2^{\Oh(k)}$ different sets of marked vertices. As the height of $D$ is at most $k$, all contracted paths had length at most $k$. Hence, each of the possible contracted trees could have been an image of contracting non-special vertices for at most $k^{\Oh(k)}$ various trees $D$, hence the total number of non-isomorphic trees $D$ with sets $M$ of $\Oh(k)$ marked vertices is at most $2^{\Oh(k)} \cdot 2^{\Oh(k)}\cdot k^{\Oh(k)} = k^{\Oh(k)}$. Each of these gives rise to at most $k^{\Oh(k)}$ non-isomorphic pairs $(D, \tau)$. 
\end{proof}

Combining \Cref{lem:promising-sets} and \Cref{lem:trees-shapes} we get the following corollary:

\begin{corollary} \label{lem:all-promising}
For a node $t$, there are only $k^{\Oh(k)}$ promising partial solutions $(D, \tau, S)$.
\end{corollary}

Also, as for any partial solution $(D, \tau, S)$, we have that $D$ is a sum of at most $|\beta(t)| \le 2k+1$ paths of length at most $k$ from its root, we get the following observation:

\begin{observation} \label{obs:partial-size}
For any partial solution $(D, \tau, S)$ for any node $t$, we have $|V(D)| = \Oh(k^2)$.
\end{observation}

What remains to be done is to understand the transitions that need to be considered in different types of bags.

\begin{description}[leftmargin=0pt]
    \item[Introduce vertex.] Let $t$ be an introduce vertex $v$ node with a child $t'$, let $(D', \tau', S')$ be a promising partial solution for $t'$ and let $(D, \tau, S)$ be a hypothetical promising partial solution for $t$ such that there exists a full solution $(F, \rho)$ such that $(D, \tau, S)$ and $(D', \tau', S')$ both extend to $(F, \rho)$. We have that $\tau|_{V(H)} = \rho|_{V(H)}$ and $\tau'|_{V(H)} = \rho|_{V(H)}$, so $\tau|_{V(H)} = \tau'|_{V(H)}$.
    As $D' = \mathsf{Clos}_F(\tau'(\beta(t')))$, $D = \mathsf{Clos}_F(\tau(\beta(t)))$ and $\beta(t) = \beta(t') \cup \{v\}$, we have that $V(D) = V(D') \cup V(P_{r \tau(v)})$, where $r$ denotes the root of $F$ and $P_{r \tau(v)}$ is the path 
    between $r$
    and $\tau(v)$ in $F$ (we remind that $\tau(v) = \rho(v)$). 
    If this path is not fully contained within $D'$, there is a~unique deepest vertex $u$ of $D'$ such that all vertices of $V(P_{r \tau(v)}) \setminus V(D')$ are descendants of $u$, or in other words $V(P_{r \tau(v)}) \setminus V(D')$ induces a path in $F$ appended at some vertex $u \in V(D')$. As $|V(D')| = \Oh(k^2)$, there are at most $\Oh(k^2)$ choices for $u$ and there are at most $k$ choices for the length of $P_{r \tau(v)}$, which gives rise to at most $\Oh(k^3)$ possibilities for $D$. In transitions where there is any vertex added (that is, if $V(P_{r \tau(v)}) \setminus V(D') \neq \emptyset$), all added vertices that are not $\tau(v)$ cannot belong to $S$ and cannot be images of $\tau$, so $\tau$ and $S$ are uniquely determined, and that in fact gives rise to at most $\Oh(k^3)$ partial solutions $(D, \tau, S)$. In transitions where $P_{r \tau(v)}$ is fully contained within $D'$, $\tau(v)$ has to belong to $V(D')$ and the choice of it uniquely determines such partial solution $(D, \tau, S)$ which gives rise to at most $\Oh(k^2)$ promising partial solutions (note that $\tau$ has to stay injective and satisfy $\tau(v) \not\in S'$). For all of these at most $\Oh(k^3)+\Oh(k^2) = \Oh(k^3)$ candidate promising partial solutions, we check if they are indeed promising, and if yes, we add them to the set of promising partial solutions for $t$. Note that the procedure generating these candidates did not depend on the choice of $(F, \rho)$.

    In the special case where $t$ is a leaf, we create promising partial solutions $(D, \tau, S)$, where $D$ is a rooted path on at most $k$ vertices, whose leaf is $\tau(v)$ and $S = \emptyset$.
    
    \item[Forget vertex.] Let $t$ be a forget vertex $v$ node with a child $t'$, let $(D', \tau', S')$ be a promising partial solution for $t'$ and let $(D, \tau, S)$ be a hypothetical partial solution for $t$ such that there exists a full solution $(F, \rho)$ such that both $(D, \tau, S)$ and $(D', \tau', S')$ extend to $(F, \rho)$
    (we will see that there is at most one such $(D, \tau, S)$ and that it does not depend on the choice of $(F, \rho)$). By forgetting $v$, we have to add $\tau(v)$ to $S'$ and remove $v$ from the domain of $\tau$, that is, $S \coloneqq S' \cup \{\tau'(v)\}$ and $\tau = \tau'|_{\beta(t)}$. In case $\tau'(v)$ was a leaf of $D'$, $D$ will have a smaller vertex set than $D'$ given by $D = \mathsf{Clos}_{D'}(\tau(\beta(t)))$ (and if it was not, then $D = D'$). Note that the restriction from $D'$ to $\mathsf{Clos}_{D'}(\tau(\beta(t)))$ may incur corresponding restriction of $S$ to $S \cap V(\mathsf{Clos}_{D'}(\tau(\beta(t))))$. However, if for $D$ defined that way we have that $V(D') \setminus V(D) \not\subseteq S' \cup \{\tau'(v)\}$, then such transition is illegal as we remove a vertex that was never assigned its preimage. We also have to additionally check that such defined $(D, \tau, S)$ is promising. If none of these checks failed, we add $(D, \tau, S)$ to the set of promising partial solutions for $t$.
    
    \item[Join.]
        Let $t$ be a join node with children $t_1$ and $t_2$. We consider all pairs $(D_1, \tau_1, S_1)$ and $(D_2, \tau_2, S_2)$ of promising partial solutions for $t_1$ and $t_2$, respectively. If $(D_1, \tau_1)$ is not isomorphic to $(D_2, \tau_2)$ we skip that pair, otherwise they can be thought of having the same set of vertices. If $S_1 \cap S_2 = \emptyset$, we add $(D_1, \tau_1, S_1 \cup S_2)$ to the set of promising partial solutions for $t$. It is easy to see that such partial solution will always be promising.
\end{description}

After we compute sets of promising partial solutions for each node $t$ according to the above rules, we check if there exists a partial solution $(D, \tau, S)$ for the root $r$ of $T$ such that $V(D) = \tau(\beta(r)) \cup S$. If yes, we output that $\MinHFunc(G) \le k$. If needed, the certifying full solution can be restored by investigating the tree of transitions leading to that particular partial solution.
All of these transitions can be realized in $k^{\Oh(1)}$ time per transition. There are at most $k^{\Oh(k)}$ transitions considered per a node $t$, 
hence the computation of the set of promising partial solutions can be done in $k^{\Oh(k)}$ time per node $t$. Therefore, the whole algorithm takes $k^{\Oh(k)} \cdot n$ time.

\newcommand\msoDFSTree{\mso{DFSTree}\xspace}
\section{MSO Formulation for DualMinHLT on general graphs}\label{sectn:courcelle-dualhlt}
\begin{lemma}\label{lem:msoDFSTree}
    There is an \MSOtwo-predicate $\msoDFSTree(T, r)$ of constant length
    which asserts that a set of edges $F$ is a DFS tree rooted at a vertex $r$.
\end{lemma}
\begin{proof}
    It is well-known (e.g.,~\cite{BoriePT92})%
    that there are  constant-length \MSOtwo-predicates $\mso{Tree}(F)$
    and $\mso{Path}(F, s, t)$
    which assert that a set of edges $F$ forms a tree and an $(s, t)$-path, respectively.
    We build
    \begin{align*}
        \mso{SpanningTree}(F) \colon \mso{Tree}(F) \land 
        \forall x~\exists e \in F ~\mso{inc}(x, e)
    \end{align*}
    which asserts that $F$ is a spanning tree in the underlying graph.
    Next, we work towards a predicate that checks that $F$ is the set of edges in a DFS tree.
    To do so, for a given $F$ and root $r$, 
    we first 
    give a predicate $\mso{anc}_{F, r}(u, v)$
    which is true if and only if $u$ is an ancestor of $v$ in the tree $F$ rooted at $r$:
    \begin{align*}
        \mso{anc}_{F, r}(u, v) \colon \exists P \subseteq F
            (\mso{Path}(P, r, v) \land \exists e \in P~ \mso{inc}(u, e))
    \end{align*}
    Now, we can verify if $F$ is a DFS tree rooted at $r$ as follows:
    \begin{align*}
        \msoDFSTree(F, r) \colon \mso{SpanningTree}(F) \land \forall u \forall v 
        (\mso{anc}_{F, r}(u, v) \lor \mso{anc}_{F, r}(v, u))
    \end{align*}
    It is clear that the length of $\msoDFSTree(F, r)$ is bounded by a constant.
\end{proof}
\newcommand\msoProbName[1]{\ensuremath{\phi_{\mbox{\textsc{#1}}}}\xspace}
\newcommand\msoDualMinLT{\msoProbName{D.\,MinHLT}}
\newcommand\msoMinLLT{\msoProbName{MinLLT}}
\newcommand\msoMaxLLT{\msoProbName{MaxLLT}}
\begin{lemma}\label{lem:mso:problems}
    There are \MSOtwo-formulas $\msoDualMinLT$, $\msoMinLLT$, $\msoMaxLLT$, 
    each of length $\Oh(k^2)$, 
    such that for any ($n$-vertex) graph $G$:
    \begin{enumerate}
        \item\label{enum:mso:DualMinLT}
        $G \models \msoDualMinLT$ if and only if $G$ has a DFS tree of height at most $n - k$.
        \item\label{enum:mso:MinLLT} 
        $G \models \msoMinLLT$ if and only if $G$ has a DFS tree with at most $k$ leaves. 
        \item\label{enum:mso:MaxLLT} 
        $G \models \msoMaxLLT$ if and only if $G$ has a DFS tree with at least $k$ leaves. 
    \end{enumerate}
\end{lemma}
\begin{proof}
    It is well-known that there are \MSOone-predicates $|X| \le k$ and $|X| \ge k$ of length $\Oh(k^2)$ that verify whether a set $X$ has size at most $k$ and at least $k$, respectively~\cite{BoriePT92}.
    Let $F$ be an edge set variable and $r$ a vertex variable.

    \Cref{enum:mso:DualMinLT}.
    To verify if the height of $F$ rooted at $r$ is at most $n - k$,
    we check the following equivalent condition: 
    there is no path $P$ in $F$ that starts at $r$ 
    such that the number of vertices that are not incident with an edge in $P$ is at most $k$.
    We first build a predicate that checks whether a vertex set $X$ consists precisely of the vertices that are not incident with a set of edges $F'$.
    \begin{align*}
        \mso{NotIncident}(F', X)\colon &\forall v (v \in X \leftrightarrow \neg \exists e \in F'~ \mso{inc}(v, e)) \\
        \mso{Height}(F, r) \le n-k \colon 
        &\neg (\exists v \exists P \subseteq F (\mso{Path}(P, r, v) \\ 
        &\land \exists X (\mso{NotIncident}(P, X) \land |X| \le k))
    \end{align*}
    Using the predicate $\msoDFSTree$ from \Cref{lem:msoDFSTree},
    we make the following \MSOtwo-formula and observe it has length $\Oh(k^2)$: 
    $$\msoDualMinLT\colon \exists F \exists r(\msoDFSTree(F, r) \land \mso{Height}(F, r) \le n - k)$$

    For \Cref{enum:mso:MinLLT,enum:mso:MaxLLT},
    we first build a predicate that checks that a set of vertices $X$ consists precisely of the vertices that have degree one in a set of edges $F$ (i.e., if $F$ is a tree, then $X$ is the set of its leaves).
    \begin{align*}
        \mso{deg1}(F, X)\colon \forall v(v \in X \leftrightarrow \exists e_1 \in F (\mso{inc}(v, e_1) \land \neg \exists e_2 \in F(\neg e_1 = e_2 \land \mso{inc}(v, e_2))))
    \end{align*}
    Now, \Cref{enum:mso:MinLLT,enum:mso:MaxLLT} can be shown as follows,
    again using the predicate \msoDFSTree from \Cref{lem:msoDFSTree}:
    \begin{align*}
        \msoMinLLT\colon &\exists F \exists r (\msoDFSTree(F, r) \land \exists X (\mso{deg1}(F, X) \land |X| \le k)) \\
        \msoMaxLLT\colon &\exists F \exists r (\msoDFSTree(F, r) \land \exists X (\mso{deg1}(F, X) \land |X| \ge k))
    \end{align*}
    Again, we observe that both of these formulas have length $\Oh(k^2)$.
\end{proof}
We obtain the following as a direct consequence of \Cref{lem:mso:problems} 
together with Courcelle's Theorem~\cite{Courcelle90} 
and Bodlaender's Theorem~\cite{Bodlaender96}.
\begin{corollary}
    \DualMinHProb,
    \textsc{MinLLT}, and
    \textsc{MaxLLT}
    parameterized by $k$ plus 
    the treewidth of the input graph
    are fixed-parameter tractable.
\end{corollary}

\section{Conclusion}
In this work we considered the parameterized complexity of \MinHProb, 
the problem of deciding if a graph has a DFS tree of height at most $k$. 
We gave an explicit FPT algorithm for the problem running in time $k^{\Oh(k)}\cdot n$. This algorithm uses a similar strategy as the one for treedepth~\cite{Reidl14}, but we achieve an improvement in the running time compared to the natural adaptation of this algorithm, that would lead to a $2^{\Oh(k^2)}\cdot n$ running time. We further showed the problem remains NP-hard on chordal graphs, but its dual version $\DualMinHProb$ is FPT.
Finally, we also show \DualMinHProb, \MaxLLTProb and \MinLLTProb are FPT parameterized by $k$ plus treewidth. 

As our algorithm for \MinHProb uses exponential space,
future work on this question may include investigating whether the algorithm by Nadara, Pilipczuk and Smulewicz \cite{Nadara22} computing treedepth in the same time but polynomial space may generalize in a similar way.

One of the most interesting open problems in this area is the question of the parameterized complexity of \DualMinHProb,
or even the complexity of \DualMinHProb for fixed $k$.
If a DFS of a graph $G$ starting at any vertex of $G$ results in a Hamiltonian path, then $G$ is called \textit{randomly traceable}~\cite{ChartrandK1968}. Such graphs characterize the NO-instances of \DualMinHProb for $k = 1$. They are
the complete graph $K_m$ on $m$ vertices, the cycle $C_m$ on $m \ge 3$ vertices, and the complete bipartite graph $K_{m,m}$ on $2m$ vertices. Since these graphs are recognizable in linear time, the \DualMinHProb is linear-time solvable for $k = 1$.
Perhaps generalizations of such characterizations can give an opening for the study of the \DualMinHProb for fixed $k$. As the next step in that direction we state the following conjecture:

\begin{conjecture}
    The problem of determining whether an $n$-vertex graph
    has a DFS-tree of height at most $n-2$ 
    is solvable in polynomial time.
\end{conjecture}

However, we also believe that \DualMinHProb does not admit an fpt-algorithm and conjecture the following:

\begin{conjecture}
    \DualMinHProb is $W[1]$-hard.
\end{conjecture}

\paragraph*{Acknowledgements.}
Lima acknowledges support of the Independent Research Fund Denmark grant agreement number 2098-00012B. Sam acknowledges support from the Research Council of Norway grant ``Parameterized Complexity for Practical Computing (PCPC)'' (NFR, no. 274526).
Nadara acknowledges support of the Independent Research Fund Denmark grant 2020-2023 (9131-00044B)
“Dynamic Network Analysis” (while being employed in Denmark) and European
Union’s Horizon 2020 research and innovation programme, grant agreement No.
948057 — BOBR (while being employed in Poland).

\bibliographystyle{plain}
\bibliography{references}

\end{document}